\documentclass[twoside,12pt,a4paper]{article}
\usepackage[T1]{fontenc}
\usepackage{textcomp}
\usepackage[english]{babel}

\usepackage{amssymb,,charter,courier,eulervm}
\usepackage[scaled]{helvet}

\newtheorem{theorem}{Theorem}
\newtheorem{lemma}[theorem]{Lemma}
\newtheorem{corollary}[theorem]{Corollary}
\newtheorem{property}[theorem]{Property}
\newtheorem{definition}{Definition}
\newtheorem{example}{Example}

\newenvironment{proof}{\par\noindent\upshape\textbf{Proof.}~}{\qed\medskip}
\newenvironment{remark}{\noindent\textsl{Remark.}~}{\par}

\newcommand\pref[1]{\mathbf{pref}(#1)}
\newcommand\qed{\hfill $\square$\par}

\renewcommand\setminus{\smallsetminus}

\newcommand\url{}
\title{A note on Automatic \textsc{Baire} property} 
\author{\itshape{}\bfseries{}Ludwig
  Staiger\thanks{email: \ttfamily{}\textbf{staiger@informatik.uni-halle.de}}\\
  Martin-Luther-Universit\"at\ Halle-Wittenberg\\Institut f\"ur Informatik
  \\von-Seckendorff-Platz 1, D--06099 Halle (Saale), Germany} 
\date{}

\begin{document}
\maketitle{}

\thispagestyle{empty}

\begin{abstract}
  Automatic Baire property is a variant of the usual Baire property which is
  fulfilled for subsets of the Cantor space accepted by finite automata. We
  consider the family $\mathcal{A}$ of all subsets of the Cantor space having the
  Automatic Baire property. In particular we show that not all finite subsets
  have the Automatic Baire property, and that already a slight increase of the
  computational power of the accepting device may lead beyond the class
  $\mathcal{A}$.
\end{abstract}

In \cite{LATA20/finkel,ijfcs21/finkel} Finkel introduced an automata-theoretic
variant of the topological Baire property for subsets of the Cantor space. He
showed that this Automatic Baire property is valid for regular
$\omega$-languages, that is, for subsets of the Cantor space definable by
finite automata.

In this note we investigate which $\omega$-languages beyond regular ones have
the Automatic Baire property. We get a full characterisation of
$\omega$-languages of first Baire category as well as of finite
$\omega$-languages having the Automatic Baire property. In this respect,
disjunctive $\omega$-words, that is, $\omega$-words random w.r.t. finite
automata in the measure-theoretic approach (cf. \cite{dcfs/St18}) play a major
r\^ole. Here, as a tool, we use the measure-category coincidence for regular
$\omega$-languages (see \cite{eik/St76}, Theorem~3 of \cite{csl/St97},
\cite{lics/VaraccaVoe}, or Section~9.4 of \cite{jacm/VaraccaVoe}).

Moreover, we show that, besides definability by finite automata, other
computational constraints do not imply Automatic Baire property.  To this end
we derive $\omega$-languages closed or open in the topology of the Cantor
space definable by simple one-counter automata not having the Automatic Baire
property.

\section{Preliminaries}
\label{sec.intro}

\subsection{Notation}
\label{sec.notat}
We introduce the notation used throughout the paper. By
$\mathbb{N} = \{ 0,1,2,\ldots\}$ we denote the set of natural numbers. Its
elements will be usually denoted by letters $i,\dots,n$. Let $X$ be an
alphabet of cardinality $|X| \ge 2$. Then $X^*$ is the set of finite words on
$X$, including the \textit{empty word} $e$, and $X^\omega$ is the set of
infinite strings ($\omega$-words) over $X$.  Subsets of $X^*$ will be referred
to as \textit{languages} and subsets of $X^\omega$ as
\textit{$\omega$-languages}.

For $w\in X^*$ and $\eta\in X^*\cup X^\omega$ let $w \cdot{}\eta$ be their
\textit{concatenation}.  This concatenation product extends in an obvious way
to subsets $W \subseteq X^*$ and $B\subseteq X^*\cup X^\omega$. For a language
$W$ let $W^* := \bigcup_{i \in \mathbb{N}} W^i$, and
$W^\omega:=\{w_1\cdots w_i\cdots: w_i\in W\setminus \{e\}\}$ be the set of
infinite strings formed by concatenating non-empty words in $W$.  Furthermore,
$|w|$ is the \textit{length} of the word $w\in X^*$ and $\pref B$ is the set
of all finite prefixes of strings in $B\subseteq X^*\cup X^\omega$.  We shall
abbreviate $w\in \pref{\{\eta\}}\ (\eta\in X^*\cup X^\omega)$ by
$w\sqsubseteq \eta$.

An $\omega$-word $\zeta\in X^{\omega}$ is \emph{disjunctive} (or \emph{rich},
\cite{csl/St97}) if every $w\in X^{*}$ is an infix of $\zeta$, that is,
$\zeta\in \bigcap_{w\in X^*} X^{*}\cdot w\cdot X^{\omega}$, and an
$\omega$-word $\xi\in X^{\omega}$ is \emph{ultimately periodic} if there are
words $w,v\in X^{*}$ such that $\xi= w\cdot v^{\omega}=w\cdot v\cdot
v\cdots$. The $\omega$-language of all ultimately periodic $\omega$-words will
be referred to as $\mathsf{Ult}$.

\subsection{Regular $\omega$-languages}
\label{sec.regular}
As usual, a language $W\subseteq X^{*}$ is \emph{regular} if it is obtained
from finite languages via the operations union, concatenation and star. An
$\omega$-language $F\subseteq X^{\omega}$ is \emph{regular} if it is of the
form $F=\bigcup_{i=1}^{n}W_{i}\cdot V_{i}^{\omega}$ where $n\in \mathbb{N}$
and $W_{i},V_{i}\subseteq X^{*}$ are regular languages.

We assume the reader to be familiar with the basic facts of the theory of
regular languages and finite automata. 
For more details on $\omega$-languages and regular $\omega$-languages see the
books \cite{PP04,TB70} or the papers \cite{handbook/St97,Th90,ic/Wag79}.

\pagebreak{}

The following is well-known.
\begin{theorem}\label{th.regBool}
  The family of regular $\omega$-languages is a Boolean algebra, and every
  non-empty regular $\omega$-language contains an ultimately periodic
  $\omega$-word.
\end{theorem}

\subsection{The Cantor space}
\label{sec.cantor}
We consider $X^\omega$ as a topological space (Cantor space). The
\emph{closure} of $F \subseteq X^\omega$ (smallest closed set containing $F$)
is $\mathcal{C}(F) := \{\xi : \pref{\{\xi\}} \subseteq \pref F\}$. The \emph{open
  sets} in Cantor space are the $\omega$-languages of the form
$W\cdot X^\omega$.  Countable unions of closed sets are referred to as
\emph{$\Sigma_{2}$-sets}, their complements as \emph{$\Pi_{2}$-sets}. The
closure $\mathcal{C}(F)$ of a regular $\omega$-language $F \subseteq X^\omega$
is again regular (\cite{eik/St76,Tr62}).

Next we recall some further topological notions, see
\cite{book/kuratowski66,Oxtoby80}.  
An $\omega$-language $F\subseteq X^{\omega}$ is \emph{nowhere dense in
  $X^{\omega}$} if its closure $\mathcal{C}(F)$ does not contain a non-empty
open subset. This property is equivalent to the fact that for all
$v\in \pref F$ there is a $w\in X^{*}$ such that $v\cdot w\notin \pref F$. If
a regular $\omega$-language $F\subseteq X^{\omega}$ is nowhere dense then
there is a word $w\in X^{*}$ such that
$F\subseteq X^{\omega}\setminus X^{*}\cdot w\cdot X^{\omega}$ \cite{eik/St76}.

Moreover, a subset $F\subseteq X^{\omega}$ is \emph{meagre} or of \emph{first
  Baire category} if it is a countable union of nowhere dense sets.

Any subset of a nowhere dense set is nowhere dense, hence, every subset of a
meagre set is again meagre. A finite union of nowhere dense sets is nowhere
dense, and a countable union of meagre sets is meagre.

The following property is a consequence of the fact that in Cantor space no
non-empty open subset is of first Baire category.
\begin{property}\label{pty.open}
  Let $F\subseteq X^{\omega}$ be of first Baire category and
  $E\subseteq X^{\omega}$ be open. If their symmetric difference
  $F\ \Delta\ E$ is of first Baire category then $E=\emptyset$.
\end{property}

\section{Measure and Category}
\label{sec.meascat}

In this section we consider the relation between measures on Cantor space and
topological density.

For every $w\in X^*$ the ball
$w\cdot X^\omega=\bigcup_{x\in X}wx\cdot X^\omega$ is a disjoint union of its
sub-balls. Thus $\mu(w\cdot X^\omega) = \sum_{x\in X}\mu(wx\cdot X^\omega)$
for every measure $\mu$ on $X^\omega$. The \emph{support} of a measure $\mu$
on $X^\omega$, $\mathbf{supp}(\mu)$, is the smallest closed subset of
$X^\omega$ such that $\mu(\mathbf{supp}(\mu))= \mu(X^\omega)$.

As measures $\mu$ on $X^\omega$ we consider finite non-null measures
($0<\mu(X^\omega)<\infty$) having the following property that the measure of a
non-null sub-ball $wx\cdot X^\omega$ does not deviate too much from
$\mu(w\cdot X^\omega)$ (cf. \cite{csl/St97,jacm/VaraccaVoe}).
\begin{definition}[Balance condition]\label{def.bala}\upshape{} A measure
  $\mu$ on $X^{\omega}$ is referred to as \emph{balanced} (or \emph{bounded
    away from zero \upshape{\cite{jacm/VaraccaVoe}}}) provided there is a constant
  $c_\mu> 0$ depending only on $\mu$ such that for all words $w\in X^*$ and
  every $x\in X$ we have $\mu(wx\cdot X^\omega) = 0$ or
  $c_\mu\cdot \mu(w\cdot X^\omega)\le \mu(wx\cdot X^\omega)$.
\end{definition}

In the book by Oxtoby \cite{Oxtoby80} analogies between topological density
and measure, in particular, the ``duality'' between measure and category, are
discussed. The papers \cite{eik/St76,csl/St97,lics/VaraccaVoe} and
\cite{jacm/VaraccaVoe} show that for regular \mbox{$\omega$-}lan\-guages in
Cantor space measure and category coincide.

\begin{theorem}[Theorem~3 of \cite{csl/St97}]\label{th.MCreg}  
  Let $F\subseteq X^{\omega}$ be a regular $\omega$-language. Then the
  following conditions are equivalent:
  \begin{enumerate}
  \item{}No $\zeta \in F$ is a disjunctive $\omega$-word.\label{th.MCreg1}
  \item{}$F$ is of first Baire category.\label{th.MCreg2}
  \item{}For all measures $\mu$ with $\mathbf{supp}(\mu)=X^\omega$ satisfying
    the balance condition it holds $\mu(F)=0$.\label{th.MCreg3}
  \item{}There is a measure $\mu$ with $\mathbf{supp}(\mu)=X^\omega$
    satisfying the balance condition such that $\mu(F)=0$.\label{th.MCreg4}
  \end{enumerate}
\end{theorem}
Items~\ref{th.MCreg1} and \ref{th.MCreg2} of Theorem~\ref{th.MCreg} show that
the union of all regular $\omega$-languages of first Baire category
$\mathbf{R}_{0}$ is the complement of the set of disjunctive $\omega$-words
(see e.g. Korollar 8 of \cite{eik/St76}).
\begin{equation}
  \label{eq.R0}
  \mathbf{R}_{0}= \bigcup\nolimits_{w\in X^*}(X^{\omega} \setminus X^{*}\cdot
  w\cdot X^{\omega})\, 
\end{equation}
\section{Baire property and Automatic Baire property}
\label{sec.baire}

Automatic Baire property was introduced by Finkel
\cite{LATA20/finkel,ijfcs21/finkel}. Here we define this variant of the usual
Baire property and derive several of its properties. First we recall the
following (see e.g. \cite{book/kuratowski66,Oxtoby80}).
\begin{definition}\label{def.Baire}\upshape{}
  A subset $F\subseteq X^{\omega}$ has the \emph{Baire property} if there is
  an open set $E\subseteq X^{\omega}$ such that $F\ \Delta\ E$ is of first
  Baire category.
\end{definition}
\begin{theorem}
  Every Borel set of the Cantor space has the Baire property.
\end{theorem}
The Automatic Baire property requires the sets $E$ and $F\ \Delta\ E$ to be
restricted in some sense to regular $\omega$-languages.
\begin{definition}[Automatic Baire property]\label{def.ABaire}\upshape{}
  A subset $F\subseteq X^{\omega}$ has the \emph{Automatic Baire property} if
  there are regular $\omega$-languages $E, F'\subseteq X^{\omega}$ where $E$
  is open and $F'$ is a $\Sigma_{2}$-set of first Baire category such that
  \begin{equation}
    \label{eq.ABaire}
    F\ \Delta\ E\subseteq F'\,.
  \end{equation}
\end{definition}
Then it holds the following.
\begin{theorem}[\cite{LATA20/finkel,ijfcs21/finkel}]\label{th.finkel}
  Every regular $\omega$-language has the Automatic Baire property.
\end{theorem}
Next we show that in Definition~\ref{def.ABaire} the requirement that $F'$ be
a $\Sigma_{2}$-set can be dropped.
\begin{corollary}\label{c.drop}
  Let $F$ be a regular $\omega$-language of first Baire category. Then there
  is a regular $\omega$-language $F'$ of first Baire category such that
  $F\subseteq F'$ and $F'$ is a $\Sigma_{2}$-set.
\end{corollary}
\begin{proof}
  Since $F$ is regular, in view of Theorem~\ref{th.finkel} there are regular
  $\omega$-lan\-guages $E, F'\subseteq X^{\omega}$ such that $E$ is open, $F'$ a
  $\Sigma_{2}$-set of first Baire category and
  \mbox{$F\ \Delta\ E\subseteq F'$}. Now the assertion follows from
  Property~\ref{pty.open}.
\end{proof}

We derive some properties of the class $\mathcal{A}$ of all $\omega$-languages
having the Automatic Baire property. It is obvious that every
$\omega$-language which has the Automatic Baire property has also the Baire
property.
\begin{lemma}\label{l.Boole}
  $\mathcal{A}$ is a Boolean algebra.
\end{lemma}
\begin{proof}
  This follows from
  $(F_{1}\cup F_{2})\ \Delta\ (E_{1}\cup E_{2})\subseteq (F_{1}\ \Delta\
  E_{1})\cup (F_{2}\ \Delta\ E_{2})$ and
  $(X^{\omega}\setminus F)\ \Delta\ (X^{\omega}\setminus E)=F\ \Delta\ E$ and
  the fact that the union of two regular $\omega$-languages of first Baire
  category is also regular and of first Baire category.
\end{proof}

\pagebreak{}
We derive a necessary condition for sets to have the Automatic Baire property.
\begin{lemma}\label{l.ABaire}
  Let $F\ \Delta\ E\subseteq F'$ where $E\subseteq X^{\omega}$ is open and
  $F'\subseteq X^{\omega}$ a regular $\omega$-language of first Baire
  category. Then for every measure $\mu$ with support
  $\mathbf{supp}(\mu)=X^{\omega}$ satisfying the balance condition it holds
  $\mu(F)=0$ if and only if $F$ is of first Baire category.
\end{lemma}
\begin{proof}
  Let $F\ \Delta\ E \subseteq F'$ where $E$ is open and $F'$ is regular and of
  first Baire category. According to Theorem~\ref{th.MCreg} we have
  $\mu(F')=0$.

  If $\mu(F)=0$ then
  $\mu(E)=\mu(E)-\mu(F)\leq \mu(E\setminus F) \leq \mu(E\ \Delta\ F) \leq
  \mu(F')=0$ implies $E=\emptyset$. Thus $F=E\ \Delta\ F$ is of first Baire
  category.
  
  If $F$ and $E\ \Delta\ F$ are of first Baire category then
  $E\subseteq (E\ \Delta\ F)\cup F$ is also of first Baire category. Thus
  $E=\emptyset$. Consequently, $\mu(F)= \mu(E\ \Delta\ F)=0$.
\end{proof}
\begin{remark}
  Observe that in Lemma~\ref{l.ABaire} we did not use the fact that the open
  set $E$ is regular.
\end{remark}
The proof of Lemma~\ref{l.ABaire} shows also the following.
\begin{corollary}\label{c.ABaire}
  Let $F\subseteq X^{\omega}$ be of first Baire category. Then
  $F\in \mathcal{A}$ if and only if $F\subseteq F'$ for some regular
  $\omega$-language of first Baire category.
\end{corollary}
Finite $\omega$-languages in ${\cal A}$ are characterised as follows.
\begin{corollary}\label{c.finite}
  Let $F\subseteq X^{\omega}$ be finite. Then $F\in\mathcal{A}$ if and only if
  $F$ does not contain a disjunctive $\omega$-word.
\end{corollary}
\begin{proof}
  If $F$ is finite then $F$ is of first Baire category. Now
  Corollary~\ref{c.ABaire} and Theorem~\ref{th.MCreg} imply that $F$ does not
  contain a disjunctive $\omega$-word.

  If $F$ is finite and does not contain a disjunctive $\omega$-word then for
  every $\xi\in F$ there is a $w_{\xi}$ such that
  $\xi\notin X^{*}\cdot w_{\xi}\cdot X^{\omega}$. Then
  $F\subseteq \bigcup_{\xi\in F}(X^{\omega}\setminus X^{*}\cdot w_{\xi}\cdot
  X^{\omega})$ which is a regular and nowhere dense $\omega$-language.
\end{proof}
Besides finite $\omega$-languages containing disjunctive $\omega$-words,
examples of sets not satisfying the Automatic Baire property are the following
ones.
\begin{lemma}\label{l.Ult}
  If $F\subseteq X^{\omega},\ \mathsf{Ult}\subseteq F\subseteq\mathbf{R}_{0}$,
  then $F$ does not have the Automatic Baire property.
\end{lemma}
\begin{proof}Assume $F\ \Delta\ E\subseteq F'$ where $E$ is open and $F'$ is a
  regular $\omega$-language of first Baire category.
  As $F\subseteq\mathbf{R}_{0}$, the set $F$ is of
  first Baire category. Now Property~\ref{pty.open} shows that  $E=\emptyset$.

  Then $\mathsf{Ult}\subseteq F\subseteq F'$. Since $F'$ is regular, we have
  $F'=X^{\omega}$ which is not of first Baire category.
\end{proof}
\begin{corollary}\label{c.Ult}
  The family $\mathcal{A}$ is not closed under countable union.
\end{corollary}
\begin{proof}
  As
  $\mathbf{R}_{0}=\bigcup_{w\in X^{*}} (X^{\omega}\setminus X^{*}\cdot w\cdot
  X^{\omega})$ and every $\omega$-language
  $X^{\omega}\setminus X^{*}\cdot w\cdot X^{\omega}$ is regular and nowhere
  dense in $X^{\omega}$ (cf. \cite{eik/St76}), the assertion follows
  immediately.
\end{proof}

\section{Simple counter-examples}\label{sec.ctr}
In Corollary~\ref{c.finite} we have seen that there are even finite
$\omega$-languages having the Baire property but not the Automatic Baire
property. Those finite \mbox{$\omega$-languages} contain $\omega$-words
$\xi\notin \mathsf{Ult}$ and are, therefore, not context-free
(e.g. \cite{tcs/EH93,handbook/St97}), that is accepted by push-down automata.

In this part we show that also a slight increase of the computational power of
accepting devices results in open or closed $\omega$-languages not having the 
Automatic Baire property.

As measure in Cantor space we use the equidistribution. For a language
$W\subseteq X^{*}$ we set $\sigma_{X}(W):= \sum_{w\in W}|X|^{-|w|}$. Then
$\mu_{=}(W\cdot X^{\omega})= \sigma_{X}(W)$, if $W\subseteq X^{*}$
prefix-free, that is, $w\sqsubseteq v$ and $w,v\in W$ imply $w=v$.

Since $\sigma_{X}(W)$ is a rational number for regular languages
$W\subseteq X^{*}$, we have the following.
\begin{theorem}[Theorem~4.16 of \cite{takeuti01}]\label{th.takeuti}
  The measure $\mu_{=}(F)$ of a regular $\omega$-language is rational.
\end{theorem}

We define the language $V_{3}\subseteq \{a,b\}^{*}$ by the equation
$V_{3}= a \cup b\cdot V_{3}^{3}$. This lan\-guage is prefix-free and satisfies
the condition that $w\cdot a^{3\cdot |w|}\in V_{3}\cdot \{a,b\}^*$ when
$w\in \{a,b\}^*$ (cf. Proposition~1.1 of \cite{ita/St05}). Moreover, it can be
accepted by a deterministic one-counter automaton using empty-storage
acceptance (cf. Example~6.3 of \cite{handbook/ABB97}).  Accordingly, the
$\omega$-languages
$V_{3}\cdot \{a,b\}^\omega, F:=\{a,b\}^\omega\setminus V_{3}\cdot
\{a,b\}^\omega$ and $V_{3}\cdot c\cdot\{a,b,c\}^\omega$ are also accepted by
deterministic one-counter automata \cite{tcs/EH93,handbook/St97}.

Since $V_{3}$ is prefix-free, the measure of these $\omega$-languages can be
easily computed from the value $\sigma_{X}(V_{3})$ which in turn is the
minimum positive solution $t_{|X|}$ of the equation
(cf. Theorem~3.1 of \cite{ita/St05})
\begin{equation}
  \label{eq.V3}
  t=  |X|^{-1}\cdot(1+t^{3})\,. 
\end{equation}
The minimum positive solutions $t_{2}=\frac{\sqrt{5}-1}{2}<1$ and $0<t_{3}<1$
are irrational\footnote{In case of $t_{3}$ assume $t_{3}=p/q$ where $p\ne q$
  are natural numbers having no common prime divisor. Then Eq.~(\ref{eq.V3})
  yields $3\cdot p\cdot q^{2}= p^{3}+q^{3}$ which is impossible.}.

The first example presents an open $\omega$-language accepted by a
deterministic one-counter automaton not satisfying the Automatic Baire
property.
\begin{example}\label{ex.open}
  We consider the open $\omega$-language
  $F_{1}:= V_{3}\cdot c\cdot\{a,b,c\}^\omega\subseteq \{a,b,c\}^\omega$.
  Since $\mu_{=}(\{a,b\}^\omega)=0$ in $\{a,b,c\}^\omega$, we obtain
  $\mu_{=}(F_{1})=\mu_{=}(F_{1}\cup \{a,b\}^\omega)=t_{3}/3$ which is
  irrational. Observe that $F_{1}\cup \{a,b\}^\omega$ is closed.

  If $E\subseteq \{a,b,c\}^\omega$ is open and regular then
  $\mathcal{C}(E)\setminus E$ is regular and nowhere dense, hence
  $\mu_{=}(\mathcal{C}(E)\setminus E)=0$ by Theorem~\ref{th.MCreg}. Now
  according to Theorem~\ref{th.takeuti} $\mu_{=}(E)=\mu_{=}(\mathcal{C}(E))$
  is rational. Thus $\mu_{=}(F_{1})\ne \mu_{=}(E)$.

  If $\mu_{=}(F_{1})> \mu_{=}(E)=\mu_{=}(\mathcal{C}(E))$ then
  $F_{1}\setminus\mathcal{C}(E)$ is non-empty and open; if
  $\mu_{=}(E)> \mu_{=}(F_{1})=\mu_{=}(F_{1}\cup \{a,b\}^\omega)$ then
  $E\setminus (F_{1}\cup \{a,b\}^\omega)\subseteq E\setminus F_{1}$ is
  non-empty and open. In both cases $F_{1}\ \Delta\ E$ contains a non-empty
  open subset, hence $F_{1}$ cannot have the Automatic Baire property.
\end{example}
Next we present a closed $\omega$-language accepted by a deterministic
one-counter automaton not having the Automatic Baire property.
\begin{example}[Example~3 of \cite{csl/St97}]\label{ex.closed}
  Define $F_{2}= \{a,b\}^\omega\setminus V_{3}\cdot \{a,b\}^\omega$ as a
  subset of the space $X^\omega=\{a,b\}^\omega$.  Then $F_{2}$ is closed and
  has, according to the value of $t_{2}$, measure
  $\mu_{=}(F_{2})=1-t_{2}= \frac{3-\sqrt{5}}{2}>0$. Since
  $w\cdot a^{3\cdot |w|}\in V_{3}\cdot \{a,b\}^*\subseteq
  X^{*}\setminus\pref{F_{2}}$,
  for $w\in \{a,b\}^*$, $F_{2}$ is nowhere dense.

  The measure $\mu_{=}$ trivially satisfies the balance condition. Now
  Lemma~\ref{l.ABaire} shows that $F_{2}$ does not have the Automatic Baire
  property.
\end{example}

\section*{ORCID}

\noindent Ludwig Staiger - \url{https://orcid.org/0000-0003-3810-9303}



\pagebreak{}

\end{document}